\nonstopmode
\documentclass[a4paper,twoside]{article}     
\usepackage{amsmath,amsthm,amssymb,amsfonts}
\usepackage[numbers]{natbib}

\usepackage{graphicx}               
\usepackage{color}                  
\usepackage[latin1]{inputenc}
\usepackage[T1]{fontenc}  
\usepackage{array}

\newcommand*\samethanks[1][\value{footnote}]{\footnotemark[#1]}
\newcommand{\beq} {\begin{eqnarray*}}
\newcommand{\eeq} {\end{eqnarray*}}
\newcommand{\trm} {\textrm}
\newcommand{\tbf} {\textbf}

\def \R{\mathbb{R}}

\def \E{\mathbb{E}}
\def \P{\mathbb{P}}
\def \Var{\hbox{{\rm Var}}}

\def \Cov{\hbox{{\rm Cov}}}
\def \myS{S_{\mathrm{Cl}}}

\newcolumntype{M}[1]{>{\raggedright}m{#1}}\usepackage{array}
\newcommand{\1}{{1\hspace{-0.2ex}\rule{0.12ex}{1.61ex}\hspace{0.5ex}}}

\def\alert#1{{\color{red} #1}}

\newtheorem{theorem}{Theorem}[section]
\newtheorem{theo}{Theorem}

\newtheorem{cor}[theorem]{Corollary}
\newtheorem{lemma}[theorem]{Lemma}
\newtheorem{rem}[theorem]{Remark}

\RequirePackage[colorlinks=true,citecolor=blue,linkcolor=blue]{hyperref}
\usepackage{hyperref}

\title{Statistical inference for Sobol pick freeze Monte Carlo method}
\author{F. Gamboa\thanks{Institut de math\'ematiques de Toulouse, Universit\'e Toulouse 3},  A. Janon\thanks{Laboratoire de Sciences Actuarielle et Financière, ISFA, Universit\'e Claude Bernard Lyon 1}, T. Klein\samethanks[1], A. Lagnoux\samethanks[1], C. Prieur\thanks{Laboratoire Jean Kuntzmann, MOISE/INRIA, Université Joseph Fourier, Grenoble}}

\begin{document}

\begin{titlepage}
\maketitle
\begin{abstract}
Many mathematical models involve input parameters, which are not precisely known. Global sensitivity analysis aims to identify the parameters whose uncertainty has the largest impact on the variability of a quantity of interest (output of the model). One of the statistical tools used to quantify the influence of each input variable on the output is the Sobol sensitivity index. We consider the statistical estimation of this index from a finite sample of model outputs. We study asymptotic and non-asymptotic properties of two estimators of Sobol indices. These properties are applied to significance tests and estimation by confidence intervals.
\end{abstract}

\tableofcontents
\end{titlepage}

\section{Introduction}
Many mathematical models encountered in applied sciences involve a large number of poorly-known parameters as inputs. It is important for the practitioner to assess the impact of this uncertainty on the model output. An aspect of this assessment is sensitivity analysis, which aims to identify the most sensitive parameters, that is, parameters having the largest influence on the output. In global stochastic sensitivity analysis (see for example \cite{saltelli-sensitivity} and references therein) the input variables are assumed to be  independent random variables.
Their probability distributions  account for the practitioner's belief about the input uncertainty. This turns the model output into a random variable, whose total variance can be split down into different partial variances (this is the so-called Hoeffding decomposition, see \cite{van2000asymptotic}). Each of these partial variances measures the uncertainty on the output induced by each input variable uncertainty. By considering the ratio of each partial variance to the total variance, we obtain a measure of importance for each input variable that is called the \emph{Sobol index} or \emph{sensitivity index} of the variable \cite{sobol1993}; the most sensitive parameters can then be identified and ranked as the parameters with the largest Sobol indices. 

Once the Sobol indices have been defined, the question of their effective computation or estimation remains open. In practice, one has to estimate (in a statistical sense) those indices using a finite sample (of size typically in the order of hundreds of thousands) of evaluations of model outputs \cite{helton2006survey}. Indeed, many Monte Carlo or quasi Monte Carlo approaches have been developed by the experimental sciences and engineering communities. This includes  
the Sobol pick-freeze (SPF) scheme  (see \cite{sobol1993,sobol2001global}). In SPF a Sobol index is viewed as the regression coefficient between the output of the model and its pick-freezed replication. This replication is obtained by holding the value of the variable of interest (frozen variable) and by sampling the other variables (picked variables). The sampled replications are then combined to produce an estimator of the Sobol index. In this paper we study very deeply this Monte Carlo method in the general framework where one or more variables can be frozen. This allows to define sensitivity indices with respect to a general random input living in a probability space (groups of variables, random vectors, random processes...). 

In \cite{jaal}, the authors have studied the asymptotic behavior of two pick-freeze estimators of a single Sobol index. The results in this paper can be continued in two directions. The first direction is motivated by the fact that 
in general, so as to rank input variables according to their importance, the pratictioners jointly estimate the collection of all the first-order as well as the total Sobol indices. As these different estimators are dependent, the asymptotic marginal distributions are not fully informative, and one has to characterize the joint law of the estimators. This joint law allows, for example, to perform significance tests and comparisons between different indices, so as to rigorously rank the input variables, taking into account indices estimation errors. The second direction is motivated by the fact that asymptotic distributions are unattainable in practice, hence, non-asymptotic tools (such as concentration inequalities, and Berry-Esseen-like theorems) about the distribution of the Sobol indices estimators should be investigated. Such results will allow conservative certification for the index estimates. 

This paper is organized as follows: in Section \ref{s:pfz}, we review the Sobol pick-freeze method and give the estimators that are studied in the paper. In Section \ref{s:joint}, we prove a central limit theorem which gives the joint asymptotic distribution of any closed Sobol index \cite{saltelli-sensitivity}, which in particular can be used to explicit the asymptotic distribution of all first-order and total index estimators. We then apply this central limit theorem to significance and comparison tests on Sobol indices. Sections \ref{s:concentre} and \ref{s:berry} are dedicated to non-asymptotic studies of the distribution of a single Sobol index estimator. These two sections, respectively, give concentration inequalities and Berry-Esseen bounds. All our theoretical results are numerically illustrated on model examples.

\section{Sobol pick freeze Monte Carlo method}
\label{s:pfz}
\subsection{Black box model and Sobol indices}
In the whole paper, we consider a non necessarily linear regression model connecting an output $Y \in \R$ to independent  random input vectors $X_1, \ldots X_p$ with for $i=1, \ldots p$, $X_i$ belongs to some probability space $\mathcal{X}_i$. We denote
\begin{equation}\label{def:lien_boite_noire}
Y=f(X):=f(X_1, \ldots , X_p)
\end{equation}
where $f$ is a deterministic real valued measurable function defined on $\mathcal{X}=\mathcal{X}_1 \times \ldots  \mathcal{X}_p$. We assume that $Y$ is square integrable and non deterministic ($\Var Y \neq 0$).

Let ${\bf u}:=(u_1,\ldots,u_k)$ be $k$ subsets of $I_p:=\{1,\ldots,p\}$. The vector of closed Sobol indices  (see \cite{saltelli-sensitivity}) is then 
$$
\myS^{{\bf u}}:=\left(\frac{\Var (\mathbb{E}(Y|X_i, i \in u_1))}{\Var(Y)}, \ldots,\frac{\Var (\mathbb{E}(Y|X_i, i \in u_k))}{\Var(Y)}\right).
$$

As pointed out and discussed in the Introduction, Sobol indices are useful quantities widely used in engineering  and applied sciences in the context of prioritisation of influent input variables of a complicated computer simulation code (see for example \cite{saltelli-sensitivity}, \cite{rocquigny2008uncertainty}) and our paper gives a rigourous statistical analysis of these quantities.
Notice that considering the whole vector $\myS^{{\bf u}}$ allows estimation of asymptotic confidence regions and tests for joint significance (see Section \ref{s:joint}).

\subsection{Monte Carlo estimation of $S$: Sobol pick freeze method}

For  $X$ and for any subset $v$ of $I_p$ we define  $X^v$ by the vector such that $X^v_i=X_i$ if $i\in v$ and $X^v_i=X'_i$ if $i\notin v$ where $X'_i$ is an independent copy of $X_i$. We then set 
\[
Y^{v}:=f(X^v).
\]

The next lemma \cite[Lemma 1.2]{jaal} shows how to express $\myS^{{\bf u}}$ in terms of covariances. This will lead to a natural estimator: 

\begin{lemma}\label{lemma:cov}
For any $u \subset I_p$, one has
\begin{equation}\label{sobol_cov}
\Var (\mathbb{E}(Y|X_i, i \in u))=\Cov\left(Y,Y^{\bf u}\right)  .
\end{equation}
\end{lemma}

An estimator with a close expression has been considered in \cite{homma1996importance}.

\tbf{Notation}\\ 
From now on, we will denote $\Var (Y)$ by $V$, $\Cov(Y,Y^{\bf u})$ by $C_u$ and $\overline{Z}_N$ the empirical mean of any $N$-sample $(Z_1,\ldots,Z_N)$ of $Z$.\\\\
\tbf{A first estimation for $\myS^{{\bf u}}$.} In view of Lemma \ref{lemma:cov}, we are now able to define a first natural estimator of $\myS^{{\bf u}}$ (all sums are taken for $i$ from 1 to $N$): 

\begin{align}\label{estsobolgen}
S^{\bf u}_{N,\mathrm{Cl}}&= \left( \frac{ \frac{1}{N} \sum  Y_i   Y_i^{{u_1}} - \left(\frac{1}{N} \sum Y_i\right) \left(\frac{1}{N}\sum  Y_i^{{u_1}}\right) }{ \frac{1}{N}\sum  Y_i^2 - \left( \frac{1}{N} \sum  Y_i \right)^2 },\ldots, \frac{ \frac{1}{N} \sum  Y_i   Y_i^{{u_k}} - \left(\frac{1}{N} \sum Y_i\right) \left(\frac{1}{N}\sum  Y_i^{{u_k}}\right) }{ \frac{1}{N}\sum  Y_i^2 - \left( \frac{1}{N} \sum  Y_i \right)^2 }\right). 
 \end{align}

These estimators have been considered in \cite{homma1996importance}, where it has been showed to be practically efficient estimators.\\

\tbf{A second estimation for $\myS^{{\bf u}}$.} Since the observations consist in $(Y_i,Y_i^{u_1}, \ldots ,Y_i^{u_k})_{(1\leq i\leq N)}$, a more precise estimation of the first and second moments can be done and we are able to define a second estimator of $\myS^{{\bf u}}$ taking into account all the available information.
Define 
\begin{align*}
Z_i^{\bf u}&=\frac1{k+1}\left( Y_i+\sum_{j=1}^k  Y_i^{{u_j}}\right), \;\;\;
M_i^{\bf u}=\frac1{k+1}\left( Y_i^2+\sum_{j=1}^k  (Y_i^{{u_j}})^2\right) .
\end{align*}
The second estimator is then defined as
\begin{align}\label{esteffgen}
T^{\bf u}_{N, \mathrm{Cl}}&= \left( \frac{ \frac{1}{N} \sum  Y_i    Y_i^{{u_1}} - \left(\frac{1}{2N} \sum (Y_i+Y_i^{u_1})\right)^2 }{ \frac{1}{N}\sum  M_i^{\bf u} -  \left(\frac{1}{N} \sum Z_i^{\bf u}\right)^2 },\ldots, \frac{ \frac{1}{N} \sum  Y_i    Y_i^{{u_k}} - \left(\frac{1}{2N} \sum (Y_i+Y_i^{u_k})\right)^2 }{ \frac{1}{N}\sum  M_i^{\bf u} -  \left(\frac{1}{N} \sum Z_i^{\bf u}\right)^2 }\right) .
 \end{align}
This estimator (in the $k=1$ case) was first introduced by Monod in \cite{Monod2006} and Janon et al. studied its asymptotic properties (CLT, efficiency) in \cite{jaal}. In \cite{owen1,owen2} Owen introduces new estimators for Sobol indices and compares numerically their performances. The delta method can also be used on these pick-freeze estimators to derive their asymptotic properties.

\begin{rem}
One could use all the information available in the sample by defining the following estimator:
\[ 
\left( \frac{ \frac{1}{N} \sum  Y_i    Y_i^{{u_1}} - \left(\frac{1}{N} \sum Z_i^u\right)^2 }{ \frac{1}{N}\sum  M_i^u -  \left(\frac{1}{N} \sum Z_i^u\right)^2 },\ldots, \frac{ \frac{1}{N} \sum  Y_i    Y_i^{{u_k}} - \left(\frac{1}{N} \sum Z_i^u \right)^2 }{ \frac{1}{N}\sum  M_i^u -  \left(\frac{1}{N} \sum Z_i^u \right)^2 }\right) .
\]
However, our empirical studies show that this estimator has a larger variance than $ T^{\bf u}_{N, \mathrm{Cl}} $.
\end{rem}

\section{Joint CLT for Sobol  index estimates with applications to significance tests}
\label{s:joint}
\subsection{Main results}

\begin{theorem}
\label{CLTgen}
Assume that $\E(Y^4)<\infty$. Then:
\begin{enumerate}
\item 
 \begin{equation}\label{ChopinS}
\sqrt{N}\left(S_{N,\mathrm{Cl}}^{\bf u} - S^{\bf u}_{\mathrm{Cl}}\right)
\overset{\mathcal{L}}{\underset{N\to\infty}{\rightarrow}}\mathcal{N}_k\left(0,\Gamma_{{\bf u},S} \right)
\end{equation}
where $\Gamma_{{\bf u},S}=\left( (\Gamma_{{\bf u},S})_{l,j} \right)_{1 \leq l, j  \leq k}$ with 
\[
 (\Gamma_{{\bf u},S})_{l,j}= \frac{ \Cov (Y Y^{u_l},Y Y^{u_j})-S^{u_l}_{\mathrm{Cl}}\Cov (Y Y^{u_j},Y^2)-S^{u_j}_{\mathrm{Cl}}   \Cov(YY^{u_l},Y^2)+ S^{u_j}_{\mathrm{Cl}} S^{u_l}_{\mathrm{Cl}}\Var(Y^2) }{ \left(\Var(Y)\right)^2 }
\]
\item 
\begin{equation}\label{ChopinT}
\sqrt{N}\left(T_{N,\mathrm{Cl}}^{\bf u} - S^{\bf u}_{\mathrm{Cl}} \right)
\overset{\mathcal{L}}{\underset{N\to\infty}{\rightarrow}}\mathcal{N}_k\left(0,\Gamma_{{\bf u},T} \right)
\end{equation}
where $\Gamma_{{\bf u},T}=\left( (\Gamma_{{\bf u},T})_{l,j} \right)_{1 \leq l, j  \leq k}$ with 
\[
 (\Gamma_{{\bf u},T})_{l,j}= \frac{\Cov (Y Y^{u_l},Y Y^{u_j})-S^{u_l}_{\mathrm{Cl}}\Cov (Y Y^{u_j},M^{\bf u})-S^{u_j}_{\mathrm{Cl}}   \Cov(YY^{u_l},M^{\bf u})+ S^{u_j}_{\mathrm{Cl}} S^{u_l}_{\mathrm{Cl}}\Var(M^{\bf u})}{\left(\Var(Y)\right)^2}.
\]
\end{enumerate}

\end{theorem}
\subsection{Some particular cases}
\begin{enumerate}
\item Assume $k=p$, $u=\left(\{1\},\ldots,\{p\}\right)$ and  $\E(Y^4)<\infty$. We denote $Y_{i}^{\{j\}}$ by $Y_{i}^{j}$. Here
\[
\myS^{{\bf u}}=\left(\frac{\Var (\mathbb{E}(Y|X_1))}{\Var(Y)}, \ldots,\frac{\Var (\mathbb{E}(Y|X_p))}{\Var(Y)}\right)
\] 
and
\begin{align*}
T^{\bf u}_{N, \mathrm{Cl}}&= \left( \frac{ \frac{1}{N} \sum  Y_i    Y_i^{1} -\left(\frac{1}{2N} \sum (Y_i+Y_i^{1})\right)^2 }{ \frac{1}{N}\sum  M_i^{\bf u} -  \left(\frac{1}{N} \sum Z_i^{\bf u}\right)^2 },\ldots, \frac{ \frac{1}{N} \sum  Y_i    Y_i^{p} -\left(\frac{1}{2N} \sum (Y_i+Y_i^{p})\right)^2 }{ \frac{1}{N}\sum  M_i^{\bf u} -  \left(\frac{1}{N} \sum Z_i^{\bf u}\right)^2 }\right) .
 \end{align*}
 The CLT becomes 
 \begin{equation*}
\sqrt{N}\left(T_{N,\mathrm{Cl}}^{\bf u} - S^{\bf u}_{\mathrm{Cl}} \right)
\overset{\mathcal{L}}{\underset{N\to\infty}{\rightarrow}}\mathcal{N}_p\left(0,\Gamma_{{\bf u},T} \right)
\end{equation*}
where $\Gamma_{{\bf u},T}=\left( (\Gamma_{{\bf u},T})_{l,j} \right)_{1 \leq l, j  \leq k}$ with 
$$\left(\Var(Y)\right)^2 \, (\Gamma_{{\bf u},T})_{l,j}= \Cov (Y Y^{l},Y Y^{j})-S^{l}_{\mathrm{Cl}}\Cov (Y Y^{j},M^{\bf u})-S^{j}_{\mathrm{Cl}}   \Cov(YY^{l},M^{\bf u})+ S^{j}_{\mathrm{Cl}} S^{l}_{\mathrm{Cl}}\Var(M^{\bf u}).$$

\item We can obviously have a CLT for any index of order 2. Indeed if we take  $k=1$ and  $(i,j)\in\{1,\ldots,p\}^{2}$ with $i\not=j$ and $u=\{i,j\}$. We get $Z^{\bf u}=\frac{1}{2}\left(Y+Y^{\bf u}\right)$ and $M^{\bf u}=\frac{1}{2}\left(Y^2+(Y^{\bf u})^2\right)$; thus

$$ \myS^{\bf u}=\frac{\Var (\mathbb{E}(Y|X_i,X_{j}))}{\Var(Y)}\ \mbox{and}\  
T^{\bf u}_{N, \mathrm{Cl}}=  \frac{ \frac{1}{N} \sum  Y_i    Y_i^{\bf u} - \left(\frac{1}{2N} \sum (Y_i+Y_i^{\bf u})\right)^2 }{ \frac{1}{2N}\sum \left(Y^2+(Y^{\bf u})^2\right) -  \left(\frac{1}{2N} \sum (Y_i+Y_i^{\bf u})\right)^2 }.
$$
 The CLT becomes 
 \begin{equation*}
\sqrt{N}\left(T_{N,\mathrm{Cl}}^{\bf u} - S^{\bf u}_{\mathrm{Cl}} \right)
\overset{\mathcal{L}}{\underset{N\to\infty}{\rightarrow}}\mathcal{N}_1\left(0,\Gamma_{{\bf u},T} \right)
\end{equation*}
with 
$$\left(\Var(Y)\right)^2 \, (\Gamma_{{\bf u},T})= \Var(Y Y^{\bf u})-2S^{{\bf u}}_{\mathrm{Cl}}\Cov (Y Y^{\bf u},Y^2)+\frac{\left(S^{{\bf u}}_{\mathrm{Cl}}\right)^2}{2} \left(\Var(Y^2)+\Cov(Y^2,(Y^{\bf u})^2)\right) .$$
\item One can also straightforwardly deduce the joint distribution of the vector of all indices of order 2. For example, if $p=3$ take $k=3$ and  ${\bf u}=(\{1,2\},\{1,3\},\{2,3\})$ and apply Theorem \ref{CLTgen}.
\end{enumerate}

\subsection{Proof of Theorem \ref{CLTgen}}
Since $S_{N, \mathrm{Cl}}^{\bf u}$  and $T_{N, \mathrm{Cl}}^{\bf u}$ are invariant by any centering (translation) of the $Y_i$'s and $Y_i^{u_j}$'s for $j=1, \ldots , k$, we can simplify the next calculations translating by $\E(Y)$. For the sake of simplicity, $Y_i$ and $Y_i^{u_j}$ now denote the centered  random variables.

{\bf Proof of (\ref{ChopinS})~:}\\
Recall that
\beq
S_{N,\mathrm{Cl}}^{\bf u} - S^{\bf u}_{\mathrm{Cl}}
=\left(\frac{\frac1N \sum Y_iY_i^{u_1}-\left(\frac1N \sum Y_i\right) \left(\frac1N \sum Y_i^{u_1}\right)}{\frac1N \sum Y_i^2-(\frac1N \sum Y_i)^2}-S^{u_1}_{\mathrm{Cl}},\ldots , \frac{\frac1N \sum Y_iY_i^{u_k}-\left(\frac1N \sum Y_i\right) \left( \frac1N \sum Y_i^{u_k}\right)}{\frac1N \sum Y_i^2-(\frac1N \sum Y_i)^2}-S^{u_k}_{\mathrm{Cl}}\right).
\eeq
Let $W_i=(Y_iY_i^{u_j},j =1, \ldots ,k, Y_i,Y_i^{u_j},j=1 \ldots ,k , Y_i^2)^t$ ($i=1,\ldots$) and $g$ the mapping from $\R^{2k+2}$ to $\R^{k}$ defined by
$$g(x_1, \ldots , x_k,y,y_1,\ldots , y_k,z)=\left(\frac{x_1-yy_1}{z-y^2},\ldots ,\frac{x_k-yy_k}{z-y^2}\right).$$

Let $\Sigma$ denote the covariance matrix of $W_i$. The vectorial central limit theorem implies that
\begin{equation*}
\sqrt{N}\left(\frac1N \sum W_i-\E (W) \right)\overset{\mathcal{L}}{\underset{N\to\infty}{\rightarrow}} \mathcal{N}_{2k+2}\left(0,\Sigma\right)
\end{equation*}
We then apply the so-called Delta method \cite{van2000asymptotic} to $W$ and $g$ so that
\begin{equation*}
\sqrt{N}\left(g(\overline{W}_N)-g(\E(W))\right)\overset{\mathcal{L}}{\underset{N\to\infty}{\rightarrow}} \mathcal{N}\left(0,J_{g}(\E(W))\Sigma J_{g}(\E(W))^t\right)
\end{equation*}
with $J_{g}(\E(W))$ the Jacobian of $g$ at point $\E (W)$.

Define $(g_1, \ldots , g_k):=\varphi$. For $i=1, \ldots , k$, $j=1, \ldots , k$,
$$\left\{\begin{array}{l}
 \frac{\partial g_j}{\partial x_i} (\E(W))= \frac{1}{V} \delta_{i,j}\\
 \frac{\partial g_j}{\partial y} (\E(W))=0\\
 \frac{\partial g_j}{\partial y_i} (\E(W))=0\\
 \frac{\partial g_j}{\partial z} (\E(W))=-\frac{S^{u_j}_{\mathrm{Cl}}}{V}
\end{array}
\right.$$
with $\delta_{i,i}=1$ and $\delta_{i,j}=0$ if $i\neq j$. Thus $\Gamma_{{\bf u},S}= J_{g}(\E(W))\Sigma J_{g}(\E(W))^t$ is as stated in Theorem \ref{CLTgen}.\\ 

{\bf Proof of (\ref{ChopinT})~:}\\
The proof is similar to the one of (\ref{ChopinS}). We now define
$W_i=(Y_iY_i^{u_j},j=1, \ldots , k,Y_i,Y_i^{u_j},j=1 \ldots ,k , \overline{(Y_i^{\bf u})^2})^t$. We apply the delta method to $g$ from $\mathbb{R}^{2k+2}$ into $\mathbb{R}^k$ defined by
$$g(x_1, \ldots, x_k, y,y_1,\ldots , y_k,z)=\left(\frac{x_1-\left(\frac{y+y_1}{2}\right)^2}{z-\left(\frac{y+y_1+\ldots+y_k}{k+1}\right)^2}, \ldots , \frac{x_k-\left(\frac{y+y_k}{2}\right)^2}{z-\left(\frac{y+y_1+\ldots+y_k}{k+1}\right)^2}\right).$$
For $i=1, \ldots , k$, $j=1, \ldots , k$,
$$\left\{\begin{array}{l}
 \frac{\partial g_j}{\partial x_i}u (\E(W))= \frac{1}{V} \delta_{i,j}\\
 \frac{\partial g_j}{\partial y} (\E(W))=0\\
 \frac{\partial g_j}{\partial y_i} (\E(W))= 0\\
 \frac{\partial g_j}{\partial z} (\E(W))=-\frac{S^{u_j}_{\mathrm{Cl}}}{V}.
\end{array}
\right.$$

\subsection{Significance tests}

In order to simplify the notation we will write the vectors $ \myS^{\bf {u}}$ as column vectors.
In this section, we give a general procedure to build significance tests of level $\alpha$ and then illustrate this procedure on two examples.

Let ${\bf u}:=(u_1,\ldots,u_k)$ so that for any $i=1,\ldots,k$, $u_i$ is a subset of $I_p:=\{1,\ldots,p\}$.  Similarly, let ${\bf v}:=(v_1,\ldots,v_l)$ and  ${\bf w}:=(w_1,\ldots,w_l)$ be $l$ be so that for any $i=1,\ldots,l$, $v_i \subseteq I_p$ and $w_i \subseteq I_p$.

  Consider the following general testing problem
$$H_{0}: \myS^{\bf {u}}=0 \mathrm{\ and\ }  \myS^{\bf v}= \myS^{\bf w} \quad \textrm{against} \quad H_{1}: H_{0}\text{ is not true}.$$ 
\begin{rem}
Note that one can also test 
$$H_{0}: \myS^{\bf {u}}\leq s \quad \textrm{against} \quad H_{1}:  \myS^{\bf {u}}>s,$$ 
or 
$$H_{0}:  \myS^{\bf u}\leq \myS^{\bf v} \quad \textrm{against} \quad H_{1}:  \myS^{\bf u}>\myS^{\bf v}.$$ 
\end{rem}

Appling Theorem \ref{CLTgen} we have
\begin{equation}\label{eqtest}
G_{N}:=\sqrt N\left(
\begin{pmatrix}
 S^{\bf {u}}_{N,\mathrm{Cl}}\\
 S^{\bf {v}}_{N,\mathrm{Cl}}- S^{\bf {w}}_{N,\mathrm{Cl}}
\end{pmatrix}
-
\begin{pmatrix}
 S^{\bf {u}}_{\mathrm{Cl}}\\
  S^{\bf {v}}_{\mathrm{Cl}}- S^{\bf {w}}_{\mathrm{Cl}}
\end{pmatrix}
\right)
\overset{\mathcal{L}}{\underset{N\to\infty}{\rightarrow}}\mathcal{N}_{k+l}\left(0,\Gamma \right).
\end{equation}
Since we have an explicit expression of $\Gamma$ we may build an estimator $\Gamma_{N} $ of $\Gamma$ thanks to empirical means. Note that $\left(\Gamma_{N}\right)_{N}$ converges a.s. to $\Gamma$. Define 
\[\widetilde G_N:=\sqrt N \begin{pmatrix}
 S^{\bf {u}}_{N,\mathrm{Cl}}\\
 S^{\bf {v}}_{N,\mathrm{Cl}}- S^{\bf {w}}_{N,\mathrm{Cl}}
\end{pmatrix}. \]
Then:
\[
G_N=\widetilde G_N-\begin{pmatrix}
 S^{\bf {u}}_{\mathrm{Cl}}\\
  S^{\bf {v}}_{\mathrm{Cl}}- S^{\bf {w}}_{\mathrm{Cl}}
\end{pmatrix}.
\]

\begin{cor}
Under $H_0$, $\widetilde G_N \overset{\mathcal{L}}{\underset{N\to\infty}{\rightarrow}}\mathcal{N}_{k+l}\left(0,\Gamma \right)$.\\
Under $H_1$, $|\widetilde G_N(1)| + |\widetilde G_N(2)| \overset{a.s.}{\underset{N\to\infty}{\rightarrow}} \infty$.
\end{cor}

This corollary allows us to construct several tests. It is a well-known fact that in the case of a vectorial null hypothesis "there exists no uniformly most powerful test, not even among the unbiased tests" (see Chapter 15 in \cite{van2000asymptotic}). In practice, we return to the dimension 1 introducing a function $F:\R^{k+l}\to \R$ and testing $H_0(F): F(h)=0$ (respectively $H_1(F): F(h)\neq 0$) instead of $H_0: h=0$ (resp. $H_1: h\neq 0$). The choice of a reasonable test "depends on the alternatives at which we wish a high power". 

\begin{rem}\label{rem_transf_lin}
If we take as test statistic $T_{N}=A \widetilde G_N$ where $A$ is a linear form defined on $\R^{l+k}$, under $H_0$, $T_{N}\overset{\mathcal{L}}{\underset{N\to\infty}{\rightarrow}}\mathcal{N}\left(0,A\Gamma A' \right)$. Replacing $\Gamma$ by $\Gamma_{N}$ and using Slutsky's lemma we get
 $$(A\Gamma_{N} A')^{-1/2}T_{N}\overset{\mathcal{L}}{\underset{N\to\infty}{\rightarrow}}\mathcal{N}\left(0,1\right).$$
 Thus we
 reject $H_{0}$ if $(A\Gamma_{N} A')^{-1/2}T_{N}\geq z_{\alpha}$ where $z_{\alpha}$ is the $1-\alpha$ quantile of a standard  Gaussian random variable.\\
 One can have a similar result when $A$ is not anymore linear but only $C^{1}$ by applying the so-called Delta method.
\end{rem}

\subsubsection{Numerical applications: toy examples} 
\paragraph{Example 1}
In this  first  toy example, we compare 5 different test statistics through their power function.
Let $X=(X_1,X_2)\sim \mathcal{N}(0,I_2)$,  and  
$$Y=f(X)=\lambda_1 X_1+\lambda_1 X_2+\lambda_{2} X_{1}X_2,$$
with $2\lambda_1^{2}+\lambda_2^{2}=1$.
We consider here the following testing problem
$$H_{0}: \myS^{\bf {1}}=  \myS^{\bf 2}=\lambda_{1}^{2} =0\quad \textrm{against} \quad H_{1} :\ \lambda_{1}\neq 0.$$ 

Then, computations lead to
\begin{align*}
\Gamma(1,1)&= \Gamma(2,2)=3-2\lambda_1^{2}-11\lambda_1^{4}+24\lambda_1^{6}-24\lambda_1^{8}\\
\Gamma(2,1)&=\Gamma(1,2)=-7\lambda_1^{4}+24\lambda_1^{6}-24\lambda_1^{8}.
\end{align*}
The Gaussian limit in Theorem \ref{CLTgen} is $\mathcal N_2(0,3Id_2)$ under $H_0$ while it  is asymptotically distributed as $\mathcal N_2(0,\Gamma)$ under $H_{{1}}$.\\\\

\tbf{Test 1}: we take as test statistic  $T_{N,1}=\widetilde G_N(1)+\widetilde G_N(2)$.\\ 
Under $H_0$, $T_{N,1}\overset{\mathcal{L}}{\underset{N\to\infty}{\rightarrow}}\mathcal{N}(0,6)$ so we reject $H_0$ if $T_{N,1}>z_{\alpha}$ where $z_{\alpha}/\sqrt{6}$ is the $(1-\alpha)$ quantile of a standard Gaussian random variable.While under $H_{{1}}$, following the procedure of Remark \ref{rem_transf_lin} with $A=(1\ 1)$.
$$\left(T_{N,1}-2\sqrt{N}\lambda_1^2\right)/(2[\Gamma(1,1)+\Gamma(1,2)])^{1/2}\overset{\mathcal{L}}{\underset{N\to\infty}{\rightarrow}} \mathcal N(0,1).$$ 
It is then easy to compute the theoretical power function. In Figure \ref{testfig1} we plot this theoretical function called \begin{rm} \begin{bf}true power fct t1 \end{bf}\end{rm} and the empirical power function called  \begin{rm} \begin{bf}estimated power fct t1\end{bf}\end{rm}. To compute the empirical power function we didn't assume the knowledge of the matrix $\Gamma$ nor the one of the function $f$.\\

\tbf{Test 2}: since the Sobol indices are non negative, the testing problem is naturally unilateral. However in view of more general contexts we introduce the test statistic  $T_{N,2}=|\widetilde G_N(1)|+|\widetilde G_N(2)|$. 
We reject $H_0$ if $T_{N,2}>z_{\alpha}$ where $z_{\alpha}/\sqrt{3}$ is the $(1-\alpha)$ quantile of the random variable having
$$\frac{2}{\sqrt{\pi}}e^{-u^2/4}\Phi(u/\sqrt{2})\1_{\R_{+}}(u)$$
as density ($\Phi$ being the distribution function of a standard Gaussian random variable). Under $H_{{1}}$, the  power function of $T_{N,2}$  and the limit variance are estimated using Monte Carlo technics. In Figure \ref{testfig1} we plot this empirical power function called  \begin{rm} \begin{bf}estimated power fct t2\end{bf}\end{rm}.  \\

\tbf{Test 3}: in the same spirit, we introduce the test statistic  $T_{N,3}=|\widetilde G_N(1)+\widetilde G_N(2)|$. 
We reject $H_0$ if $T_{N,3}>z_{\alpha}$ where $z_{\alpha}/\sqrt{6}$ is the $(1-\alpha/2)$ quantile of a standard Gaussian random variable.Under $H_{{1}}$, the  power function of $T_{N,3}$  and the limit variance are estimated using Monte Carlo technics. In Figure \ref{testfig1} we plot this empirical power function called  \begin{rm} \begin{bf}estimated power fct t3\end{bf}\end{rm}.\\

\tbf{Test 4}: we use the $L^2$ norm and consider $T_{N,4}=(G_N(1))^2+(G_N(2))^2$.
Under $H_0$, $T_{N,4}/3\overset{\mathcal{L}}{\underset{N\to\infty}{\rightarrow}} \chi_2(2)$ so we reject $H_0$ if $T_{N,4}>z_{\alpha}$ where $z_{\alpha}/3$ is the $(1-\alpha)$ quantile of a $\chi_2$ random variablewith 2 degrees of freedom. Under $H_{{1}}$, the  power function of $T_{N,4}$  and the limit variance are estimated using Monte Carlo technics. We plot this empirical power function in Figure \ref{testfig1} called    \begin{rm} \begin{bf}estimated power fct t4\end{bf}\end{rm}.\\

\tbf{Test 5}: we use the infinity norm and consider $T_{N,5}=max(|G_N(1)|;|G_N(2)|)$.
We reject $H_0$ if $T_{N,5}>z_{\alpha}$ where $z_{\alpha}/\sqrt{3}$ is the $[1+\sqrt{1-\alpha}]/2$ quantile of a standard Gaussian random variable.Under $H_{{1}}$, the  power function of $T_{N,5}$  and the limit variance are estimated using Monte Carlo technics. In Figure \ref{testfig1} we plot this theoretical function called \begin{rm} \begin{bf}true power fct t5 \end{bf}\end{rm} and the empirical power function called  \begin{rm} \begin{bf}estimated power fct t5\end{bf}\end{rm}. \\

In Figure \ref{testfig1} we thus present the plot of the different power functions for $N=100,\ 500$ and $1000$. 
 
\begin{figure}[h]
\begin{center}
\includegraphics[scale=.75]{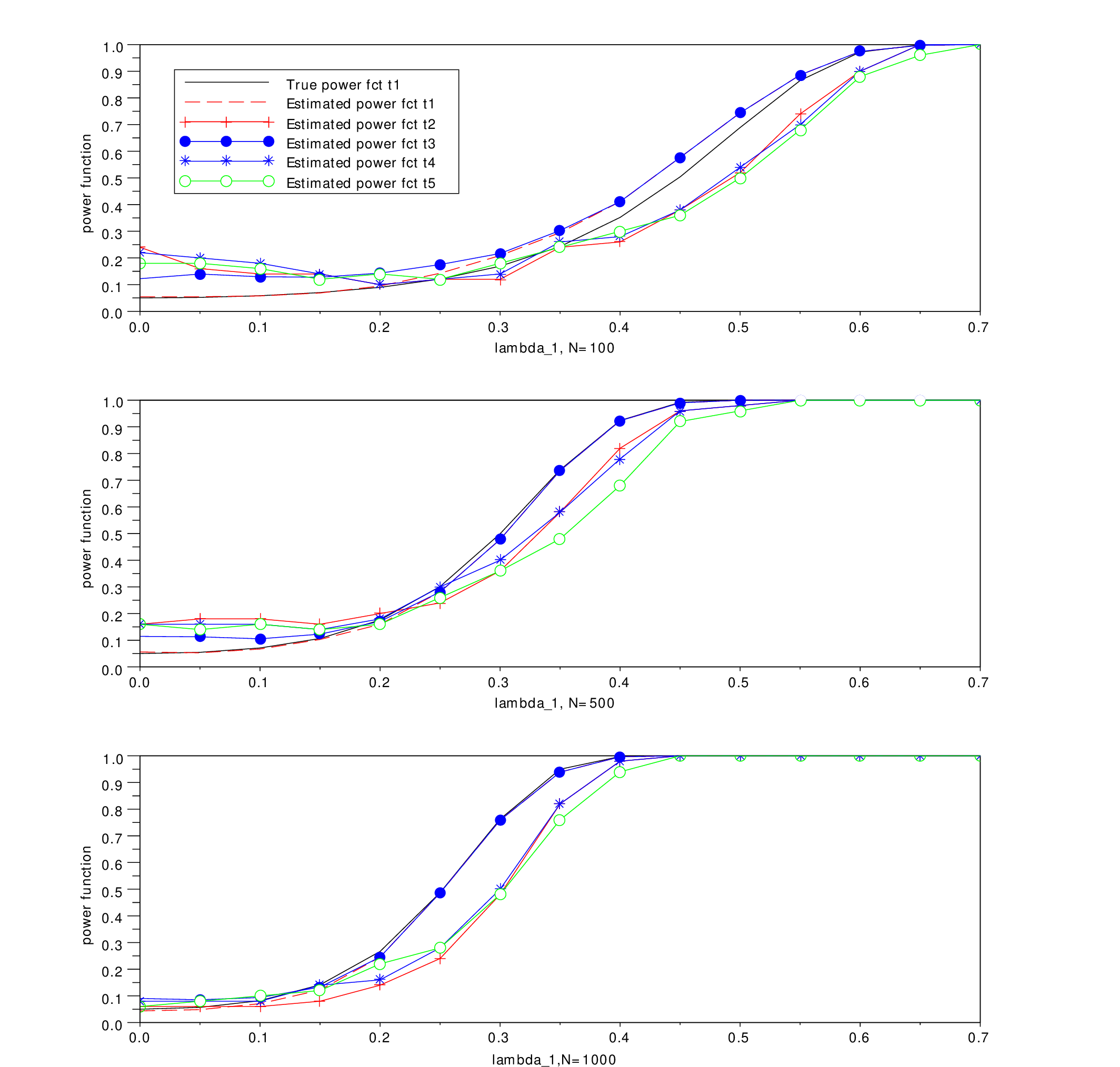}
\caption{Power functions}
\label{testfig1}
\end{center}

\end{figure}

\paragraph{Example 2}

Let $X=(X_1,X_2,X_{3})\sim \mathcal{N}(0,I_3)$, $2\lambda_{1}^{2}+\lambda_{2}^{2}=1$ and 
$$Y=f(X)=\lambda_1 (X_2+ X_{3}) +\lambda_{2}X_{1}X_{2}.$$
Let us test if $X_{1}$ has any influence ie $H_{0}: \myS^{\{\bf 1\}}=0$,  $\myS^{\bf{\{1,2\}}}=S^{\{\bf 2\}}_{\mathrm{Cl}}$ and $S^{\{\bf{1,3}\}}_{\mathrm{Cl}}=S^{\{\bf3\}}_{\mathrm{Cl}}$ . 
Applying Theorem \ref{CLTgen} we easily get 
$$G_{N}:=
\sqrt{N}\left(
\begin{pmatrix}
S_{N,\mathrm{Cl}}^{ 1}\\
S_{N,\mathrm{Cl}}^{1, 2}  -S_{N,\mathrm{Cl}}^{2} \\
S_{N,\mathrm{Cl}}^{1, 3}  -S_{N,\mathrm{Cl}}^{3}
\end{pmatrix}
-
\begin{pmatrix}
S_{\mathrm{Cl}}^{ 1}\\
S_{\mathrm{Cl}}^{1, 2}  -S_{\mathrm{Cl}}^{2} \\
S_{\mathrm{Cl}}^{1, 3}  -S_{\mathrm{Cl}}^{3}
\end{pmatrix}
\right)
\overset{\mathcal{L}}{\underset{N\to\infty}{\rightarrow}}\mathcal{N}_3\left(0,\Gamma \right).
$$
Here under $H_0$ the covariance limit $\Gamma$ in Theorem \ref{CLTgen} is the identity matrix. Under $H_1$ we use its explicit expression given in Theorem \ref{CLTgen} to compute an empirical estimator $\Gamma_{N}$. We compare Test 1, Test 3, Test 4 and Test 5 defined in the previous example. We present in Figure \ref{testfig2} the plot of the different estimated power functions for  $N=100,\ 500$ and $1000$.
  
\begin{figure}[h]\begin{center}
\includegraphics[scale=.75]{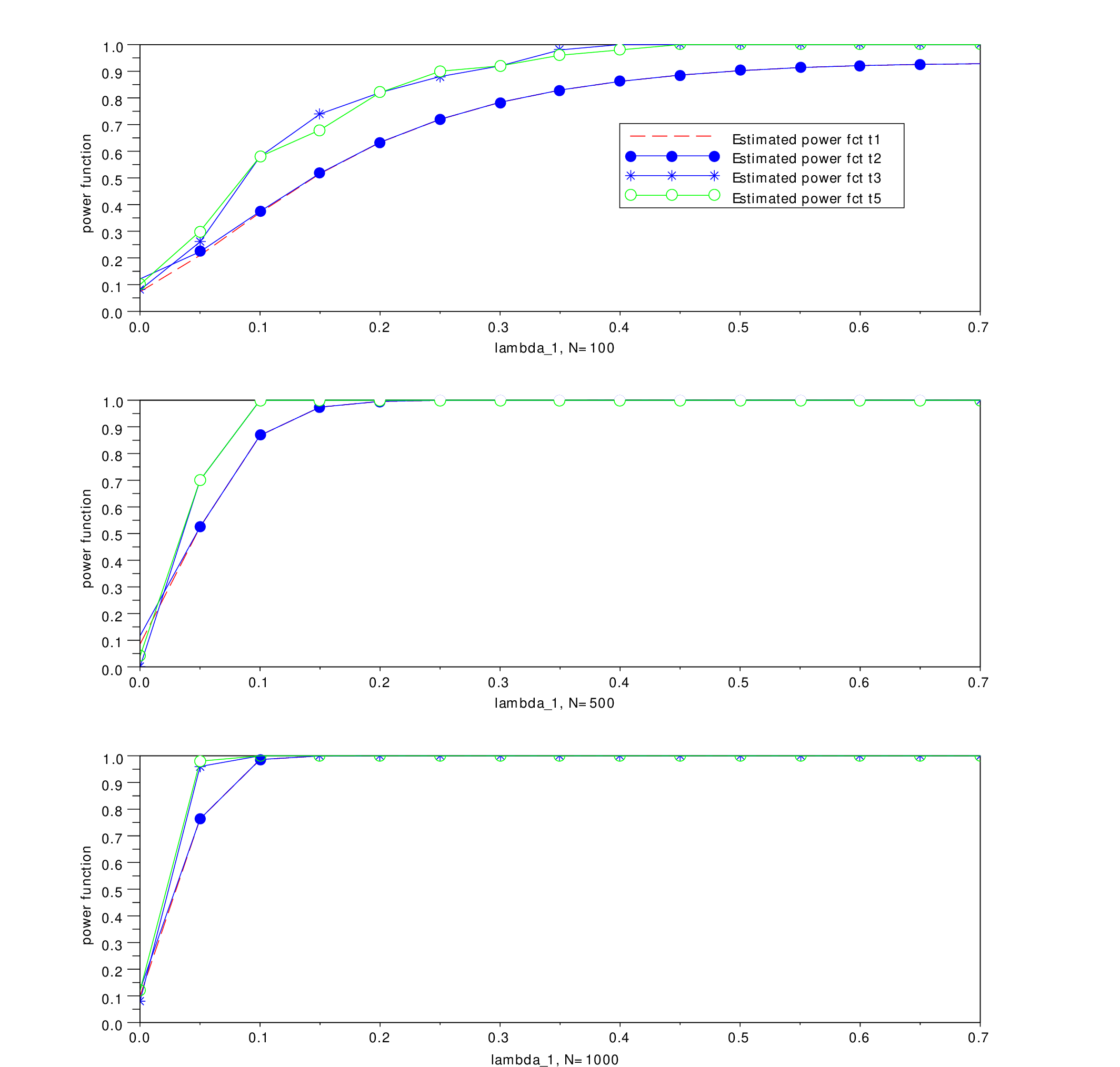}
\end{center}
\caption{Estimated power functions for different values of $N$.}
\label{testfig2}
\end{figure}

Figures \ref{testfig1} and \ref{testfig2} show, as expected, that increasing $N$ leads to a steeper power function (hence, a better discrimination between the hypothesis), and that the estimated power function gets closer to the true one. We also see that no test is the most powerful, uniformly in $\lambda_1$, in accordance with the theory quoted above.

\paragraph{Ishigami function}
\label{ss:ishigamifun}
The Ishigami model 
\cite{ishigami1990importance}
is given by:
\begin{equation}
	\label{e:ishigam}
	Y = f(X_1,X_2,X_3)=\sin X_1+7 \sin^2 X_2+0.1 X_3^4 \sin X_1 
\end{equation}
for $(X_j)_{j=1,2,3}$ are i.i.d. uniform random variables in $[-\pi;\pi]$. 
Exact values of these indices are analytically known: 
\[ \myS^{\{1\}}=0.3139, \;\;\myS^{\{2\}}=0.4424, \;\; \myS^{\{3\}}=0. \]
We  perform simulations in order to show that our test procedure allows us to recover the fact that $\myS^{\{3\}}=0$, even for relatively small values of $N$. In Table \ref{table_Ishi}, we present the simulated confidence levels obtained for $N\in\{10,50,100,500,1000\}$ by the following procedure. For each value of $N$, we use  a 1000 sample to estimate the confidence level and we repeat this scheme 20 times.  We give in Table  \ref{table_Ishi} the minimum, the mean and the maximum of these 20 distinct simulated values of the confidence levels.
\begin{table}
\begin{center}
\begin{tabular}{|c|c|c|c|}\hline
N&Min&Mean&Max\\\hline\hline
10&    0.041  &  0.0463   & 0.048    \\ \hline
50&      0.042 &   0.0482 &   0.050    \\ \hline
100&   0.044 &   0.0489  &  0.051               \\ \hline
 500&0.047 &   0.0510  &  0.053  \\ \hline
 1000& 0.049  &  0.0510&0.055  \\ \hline

\end{tabular}
\caption{Results for the Ishigami function}\label{table_Ishi}
\end{center}
\end{table}

%
\subsubsection{Numerical applications: a real test case}

It is customary in aeronautics to model the fuel mass needed to link two fixed countries with a commercial aircraft by the Bréguet formula:
\begin{eqnarray}
  M_{fuel} = \left( M_{empty} + M_{pload} \right)\left( e^{\frac{SFC\cdot g\cdot Ra}{V\cdot F}\,10^{-3}} - 1 \right)\,
  . \label{Mfuelmodel}
\end{eqnarray}
See \cite{rachdi2012stochastic} for the description of the model with more details.

 The fixed variables are
\begin{itemize}
  \item $M_{empty}$ : {\it Empty weight} = basic weight of the aircraft (excluding fuel and passengers)
  \item $M_{pload}$ : {\it Payload} = maximal carrying capacity of the aircraft
  \item $g$ : Gravitational constant
  \item $Ra$ : {\it Range} = distance traveled by the aircraft
\end{itemize}

The uncertain variables are
\begin{itemize}
  \item $V$ : {\it Cruise speed} = aircraft speed between ascent and descent phase
  \item $F$ : {\it Lift-to-drag ratio} = aerodynamic coefficient
  \item $SFC$ : {\it Specific Fuel Consumption} = characteristic value of engines
\end{itemize}

We follow \cite{rachdi2012stochastic} and model the uncertainties as presented in
Table \ref{uncertainty}.

\begin{table}[!htbp]
\centering
  \begin{tabular}{|c|c|c|}
  \hline
      variable & density &  parameter  \\
    \hline
    \hline
   $V$ & {\it Uniform} & $(V_{min},V_{max})$ \\
   \hline
   $F$ &   {\it Beta }  & $(7,2,F_{min},F_{max})$ \\
    \hline
    $SFC$ & $\theta_{2}\,e^{-\theta_{2}(u-\theta_{1})}\,\1_{[\theta_{1},+\infty [ }$ & $\theta_1=17.23,\theta_2=3.45$\\ 
    \hline
\end{tabular}
\caption{Uncertainty modeling} \label{uncertainty}
\end{table}

The probability density function of a beta distribution on $[a,b]$
with shape parameters $(\alpha,\beta)$ is
$$ g_{(\alpha,\beta,a,b)}(x) = \frac{(x-a)^{(\alpha-1)}(b-x)^{\beta-1}}{(b-a)^{\beta-1}B(\alpha,\beta)}\1_{[a,b](x)} \, ,$$
where $B(\cdot,\cdot)$ is the beta function. Still following \cite{rachdi2012stochastic}, we take the nominal and extremal values of $V$ and $F$ as in Table \ref{tableminmax}.

\begin{table}[!htbp]
\centering
  \begin{tabular}{|c|c|c|c|}
  \hline
      variable & nominal value &  min & max  \\
    \hline
    \hline
   $V$ & {\bf 231} & 226 & 234 \\
   \hline
   $F$ & {\bf 19} & 18.7 & 19.05 \\
    \hline
\end{tabular}
\caption{Minimal and maximal values of uncertain variables}
\label{tableminmax}
\end{table}

The uncertainty on the cruise speed $V$ represents a relative
difference of
arrival time of $8$ minutes.

The airplane manufacturer may wonder whether he has to improve the quality of the engine ($SFC$) or the aerodynamical property of the plane ($F$). Thus we study the sensitivity of $M_{fuel}$ with respect to $F$ and $SFC$ and we want to know if $H_{0}:\ S^{SFC}>S^{F}$ or $H_{1}:\ S^{SFC}\leq S^{F}$. Applying the test procedure described previously we can not reject $H_{0}$.

\section{Concentration inequalities}
\label{s:concentre}
In this section we give concentration inequalities satisfied by the Sobol indices in one dimension (i.e. $k=1$). 

We define the $h$ function by $h(x)=(1+x)\ln(1+x)-x$ for all $x>-1$.

%
%
%

\subsection{Concentration inequalities for $S^{\bf u}_{N,\mathrm{Cl}}$}
\label{ss:concentreS}

We introduce the random variables 
$$U_{i}^ {\pm}=Y_iY_i^{\bf u}-(\myS^{{\bf u}}\pm y)(Y_i)^2\; \trm{and} \; 
J_{i}^{\pm}=(\myS^{{\bf u}}\pm y)Y_i-Y_i^{\bf u}$$

and denote $V_U^+$ (respectively $V_U^-$, $V_J^+$ and $V_J^-$) the second moment of the i.i.d. random variable $U_{i}^ +$  (resp.  $U_{i}^{-}$, $J_i^+$ and $J_i^-$).

\begin{theo}
\label{thm:concentreS}
Let $b>0$ and $y > 0$. Assume that all the random variables $Y_{i}$ and $Y^{\bf u}_{i}$ belong to $[-b,b]$. Then
\begin{align}
\P\left(S^{\bf u}_{N,\mathrm{Cl}}\geq \myS^{{\bf u}}+y\right)&\leq  M_1+2M_2+2M_3,\label{concentrationR_nc}\\
\P\left(S^{\bf u}_{N,\mathrm{Cl}}\leq \myS^{{\bf u}}-y\right)&\leq M_4+2M_2+2M_5,\label{concentrationL_nc}
\end{align}
where
\[
 \begin{array}{l}
 \displaystyle
M_1=\exp\left\{-\frac{N V_U^+}{b_{U}^{2}}h\left(\frac{b_{U}}{V_U^+}\frac{yV}{2}\right)\right\}\\
 \displaystyle
M_2=\exp\left\{-\frac{N V}{b^{2}}h\left(\frac{b}{V}\sqrt{\frac{yV}{2}}\right)\right\}
 \end{array}
 \;\;
 \begin{array}{l}
 \displaystyle
M_3=\exp\left\{-\frac{N V_J^+b^2}{b_{U}^{2}}h\left(\frac{b_{U}}{bV_J^+}\sqrt{\frac{yV}{2}}\right)\right\} \\
 \displaystyle
M_4=\exp\left\{-\frac{N V_U^-}{b_{U}^{2}}h\left(\frac{b_{U}}{V_U^-}\frac{yV}{2}\right)\right\}
\end{array}
\]
\[
M_5=\exp\left\{-\frac{N V_J^-b^2}{b_{U}^{2}}h\left(\frac{b_{U}}{bV_J^-}\sqrt{\frac{yV}{2}}\right)\right\}
\]
and $b_{U}=b^{2}(1+\myS^{{\bf u}}+y)$.
\end{theo}

\begin{rem}
One must be cautious since the variables $Y_i-\overline{Y}_N$ are dependent.
\end{rem}

\begin{proof}  Since $\myS^{{\bf u}}$ and $S^{\bf u}_{N,\mathrm{Cl}}$ are invariant by translation on $Y$ and $Y^{\bf u}$, one may assume without loss of generality that $Y$ is centered.
\begin{enumerate}
\item Obviously $U_i^+$ and $U_i^-$ are upper-bounded by $b_U$, $J_i^+$ and $J_i^-$ by $b_U/b$,
\begin{displaymath}
\begin{array}{lll}
\E(U_i^+)=-yV && \E(J_i^+)=0\\
\E(U_i^-)=yV && \E(J_i^-)=0\\
\end{array}
\end{displaymath}
and
\begin{displaymath}
\begin{array}{l}
V_U^{\pm} = \Var(YY^{\bf u})+(\myS^{{\bf u}}+y)^2\Var(Y^2)-2(\myS^{{\bf u}}\pm y)\Cov(YY^{\bf u},Y^2)+y^2V^2\\
\\
V_J^{\pm} = ((\myS^{{\bf u}}\pm y)^2+1)V-2(\myS^{{\bf u}}\pm y)C_u.\\
\end{array}
\end{displaymath}

\item Proof of (\ref{concentrationR_nc}). Using 
$$\{a+b\geq c\}\subset \{a\geq c/2\} \cup \{b\geq c/2\} \quad \trm{and} \quad \{ab\geq c\}\subset \{|a|\geq \sqrt c\} \cup \{|b|\geq \sqrt c\}$$

one gets
\begin{eqnarray*}
\P\left(S^{\bf u}_{N,\mathrm{Cl}}\geq \myS^{{\bf u}}+y\right)&=&\P\left(\frac{\frac{1}{N}\sum_{i=1}^N Y_iY_i^{\bf u}-\overline{Y}_N\overline{Y}^{\bf u}_N}{\frac{1}{N}\sum_{i=1}^N  (Y_i)^2-\left(\overline{Y}_N\right)^2}\geq \myS^{{\bf u}}+y\right) \\
&=&\P\left(\frac{1}{N}\sum_{i=1}^N \left(U_i^+-\E(U^+)\right)+\overline{Y}_N \overline{J}^+_N\geq yV\right) \\
&\leq & \P\left(\sum_{i=1}^N \left(U_i^+-\E(U^+)\right)]\geq N\frac{yV}{2}\right) +\P\left(\sum_{i=1}^N Y_i\geq N\sqrt{\frac{yV}{2}}\right)\\
&&+ \P\left(\sum_{i=1}^N (-Y_i)\geq N\sqrt{\frac{yV}{2}}\right)
+\P\left(\sum_{i=1}^N J_i^+\geq N\sqrt{\frac{yV}{2}}\right) \\
&&+ \P\left(\sum_{i=1}^N (-J_i^+)\geq N\sqrt{\frac{yV}{2}}\right).\\
\end{eqnarray*}

Inequality (\ref{concentrationR_nc}) comes directly by applying five times Bennett inequality (see \cite{boucheron2013concentration} and references therein).\\

\item  Proof of (\ref{concentrationL_nc}). In the same way,
one gets
\begin{eqnarray*}
\P\left(S^{\bf u}_{N,\mathrm{Cl}}\leq \myS^{{\bf u}}-y\right)
&=&\P\left(\frac{1}{N}\sum_{i=1}^N \left(-U_i^-+\E(U^-)\right)+(-\overline{Y}_N) \overline{J}^-_N\geq yV\right) \\
&\leq & \P\left(\sum_{i=1}^N \left(-U_i^-+\E(U^-)\right)\geq N\frac{yV}{2}\right) +\P\left(\sum_{i=1}^N Y_i\geq N\sqrt{\frac{yV}{2}}\right)\\
&&+ \P\left(\sum_{i=1}^N (-Y_i)\geq N\sqrt{\frac{yV}{2}}\right)
+\P\left(\sum_{i=1}^N J_i^-\geq N\sqrt{\frac{yV}{2}}\right)\\
&&+ \P\left(\sum_{i=1}^N (-J_i^-)\geq N\sqrt{\frac{yV}{2}}\right).\\
\end{eqnarray*}

Inequality (\ref{concentrationL_nc}) comes directly by applying five times Bennett inequality. \qedhere
\end{enumerate} 
\end{proof}

\subsection{Concentration inequalities for $T^{\bf u}_{N,\mathrm{Cl}}$}
\label{ss:concentreT}

Now remind $Z_{i}=\frac{Y_i+Y_i^{\bf u}}{2}$ and introduce the random variables $K_{i}^ {\pm}=Y_iY_i^{\bf u}-(\myS^{{\bf u}}\pm y)\frac{(Y_i)^2+(Y_i^{\bf u})^2}{2}.$

Denote $V_{K}^+$ (resp. $V_{K}^-$) the second moment of the i.i.d. random variable $K_{i}^ +$  (resp.  $K_{i}^{-}$).

\begin{theo}
\label{thm:concentreT}
Let $b>0$ and $y>0$. Assume that $Y \in [-b,b]$. Then
\begin{align}
\P\left(T^{\bf u}_{N,\mathrm{Cl}}\geq \myS^{{\bf u}}+y\right)&\leq  m_1+2m_2\1_{\{\myS^{{\bf u}}+y-1\geq 0\}},\label{concentrationR}\\
\P\left(T^{\bf u}_{N,\mathrm{Cl}}\leq \myS^{{\bf u}}-y\right)&\leq m_3+2m_4 \1_{\{\myS^{{\bf u}}+y-1\geq 0\}},\label{concentrationL}
\end{align}
where
\[
m_1=\exp\left\{-\frac{N V_{K}^+}{b_{U}^{2}}h\left(\frac{b_{U}}{V_{K}^+}\frac{yV}{2}\right)\right\},\;\;\;
m_2=\exp\left\{-\frac{N (V+C)}{2b^{2}}h\left(\frac{b}{V+C}\sqrt{\frac{2yV}{\myS^{{\bf u}}+y-1}}\right)\right\},
\]
\[
m_3=\exp\left\{-\frac{N V_{K}^-}{b_{U}^{2}}h\left(\frac{b_{U}}{V_{K}^-}\frac{yV}{2}\right)\right\},\;\;\;
m_4=\exp\left\{-\frac{N (V+C)}{2b^{2}}h\left(\frac{b}{V+C}\sqrt{\frac{2yV}{y+1-\myS^{{\bf u}}}}\right)\right\}.
\]
\end{theo}

\begin{proof}  Since  $T^{\bf u}_{N,\mathrm{Cl}}$ is invariant by translation on $Y$ and $Y^{\bf u}$, one may assume without loss of generality that $Y$ is centered.
\begin{enumerate}
\item Obvisouly $K_i^+$ and $K_i^-$ are upper-bounded by $b_U$, $\E(K_i^+)=-yV$, $\E(K_i^-)=yV$ and
$$V_{K}^{\pm} = V_{U}^{\pm}+(\myS^{{\bf u}}\pm y)^2\frac{\Cov(Y^2,(Y^{\bf u})^2)-\Var(Y^2)}{2}.$$
We also have $Z_i$ is upper-bounded by $b$, $\E(Z_i)=0$ and $\E(Z_i^2)=\frac{V+C_u}{2}$.
\item Proof of (\ref{concentrationR}). One gets if $\myS^{{\bf u}}+y-1\geq 0$ 
\begin{eqnarray*}
\P\left(T^{\bf u}_{N,\mathrm{Cl}}\geq \myS^{{\bf u}}+y\right)&=&\P\left(\frac{\frac1N\sum_{i=1}^N Y_iY_i^{\bf u}-\left(\overline{Z}_N\right)^2  }{\frac1N\sum_{i=1}^N  \frac{Y_i^2+(Y_i^{\bf u})^2}{2}-\left(\overline{Z}_N\right)^2}\geq \myS^{{\bf u}}+y\right) \\
&=&\P\left(\frac{1}{N}\sum_{i=1}^N \left(K_i^+-\E(K^+)\right)+\left(\overline{Z}_N\right)^2(\myS^{{\bf u}}+y-1)\geq yV\right) \\
&\leq & \P\left(\sum_{i=1}^N \left(K_i^+-\E(K^+)\right)\geq N\frac{yV}{2}\right) +\P\left(\sum_{i=1}^N Z_i\geq N\sqrt{\frac{yV}{2(\myS^{{\bf u}}+y-1)}}\right)\\
&&+\P\left(\sum_{i=1}^N (-Z_i) \geq N\sqrt{\frac{yV}{2(\myS^{{\bf u}}+y-1)}}\right).
\end{eqnarray*}

Inequality (\ref{concentrationR}) comes directly by applying Bennett inequality to the random variables $K_i^+$, $Z_i$ and $-Z_i$.\\

\item  Proof of (\ref{concentrationL}). One gets since $y+1-\myS^{{\bf u}}>0$
\beq
\P\left(T^{\bf u}_{N,\mathrm{Cl}}\leq \myS^{{\bf u}}-y\right)&=&\P\left(\frac{1}{N}\sum_{i=1}^N \left(-K_i^-+\E(K^-)\right)+\left(\overline{Z}_N\right)^2(y+1-\myS^{{\bf u}})\geq yV\right) \\
&\leq & \P\left(\sum_{i=1}^N \left(-K_i^-+\E(K^-)\right)\geq N\frac{yV}{2}\right) +\P\left(\sum_{i=1}^N Z_i\geq N\sqrt{\frac{yV}{2(y+1-\myS^{{\bf u}})}}\right)\\
&&+\P\left(\sum_{i=1}^N (-Z_i)\geq N\sqrt{\frac{yV}{2(y+1-\myS^{{\bf u}})}}\right).
\eeq

Inequality (\ref{concentrationL}) comes from Bennett inequality to the random variables $K_i^-$, $Z_i$ and $-Z_i$.\qedhere
\end{enumerate} 
\end{proof}

\subsection{Numerical applications}
\label{ss:numconcentre}
In this section, we provide numerical illustrations of the concentration inequalities stated in Sections \ref{ss:concentreS} and \ref{ss:concentreT}.

The upper bounds appearing in Theorem \ref{thm:concentreS} involve the (a priori) unknown quantities:
\[ Q = \left( V, V_U^+, V_U^-, V_J^+, V_J^-, \myS^{\bf u} \right). \]
We denote by $pAbove(y,N)$ and $pBelow(y,N)$ the estimators of the right-hand sides of \eqref{concentrationR_nc} and \eqref{concentrationL_nc}, respectively, obtained by replacing the $Q$ vector by its empirical estimate.

Similarly, we denote by $pAbove'(y,N)$ and $pBelow'(y,N)$ the estimators of the right-hand sides of \eqref{concentrationR} and \eqref{concentrationL} when:
\[ Q' = \left( V, C, V_K^+, V_K^-, \myS^{\bf u}    \right) \]
is replaced by its empirical estimate.

One should note at this point that, on the one hand, the bounds of Theorems \ref{thm:concentreS} and \ref{thm:concentreT} are fully rigorous for any $N$. From a practical point of view, these bounds are not computable, unless the $Q$ (resp. $Q'$) vector is known. On the other hand, $pAbove$ and $pBelow$ (resp. $pAbove'$ and $pBelow'$) are computable but are not fully justified for finite $N$, as they rely on the estimation of $Q$ (resp. $Q'$). However, as pointed out in \cite{hickernell2012guaranteed}, these bounds are \emph{conservative}, hence they are less sensitive to a bad estimation than the asymptotic confidence interval given by the CLT

We again take for $f$ the Ishigami function considered in \ref{ss:ishigamifun}.

In this case, it is easy to check that $Y \in [-b, b]$, where:
\[ b=8 + 0.1 \times \pi^4. \]
When such a majoration of $Y$ is not possible, $b$ can be put into the $Q$ (or $Q'$) vector and estimator of it can be plugged in to obtain $pAbove$ and $pBelow$ (or $pAbove'$ and $pBelow'$).

We also choose ${\bf u}=\{1\}$.

Figure \ref{fig:concentre} show, for different values of $N$, the plot of $pAbove(y,N)$ and $pBelow(y,N)$ (respectively, $pAbove'(y,N)$ and $pBelow'(y,N)$) as functions of $y$.

\begin{figure}[h]
	\begin{center}
		\includegraphics[scale=.6]{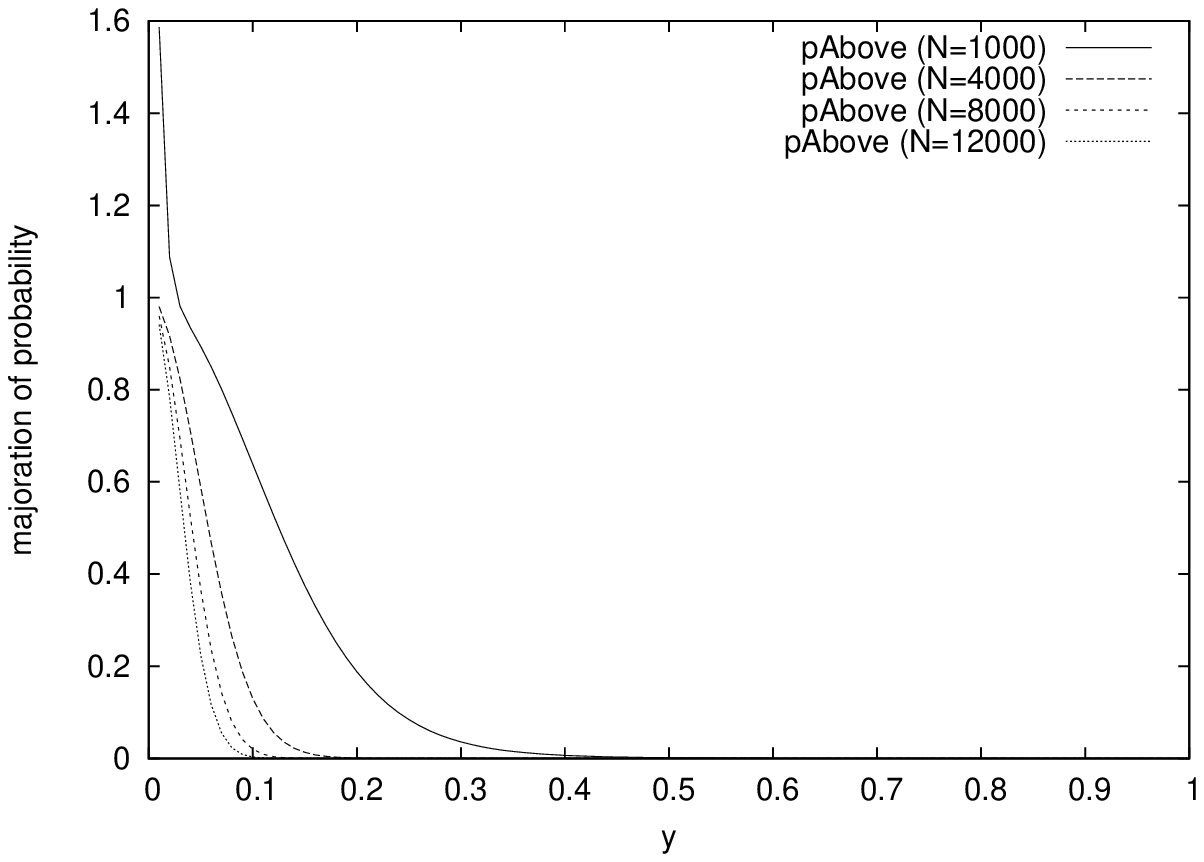}
		\includegraphics[scale=.6]{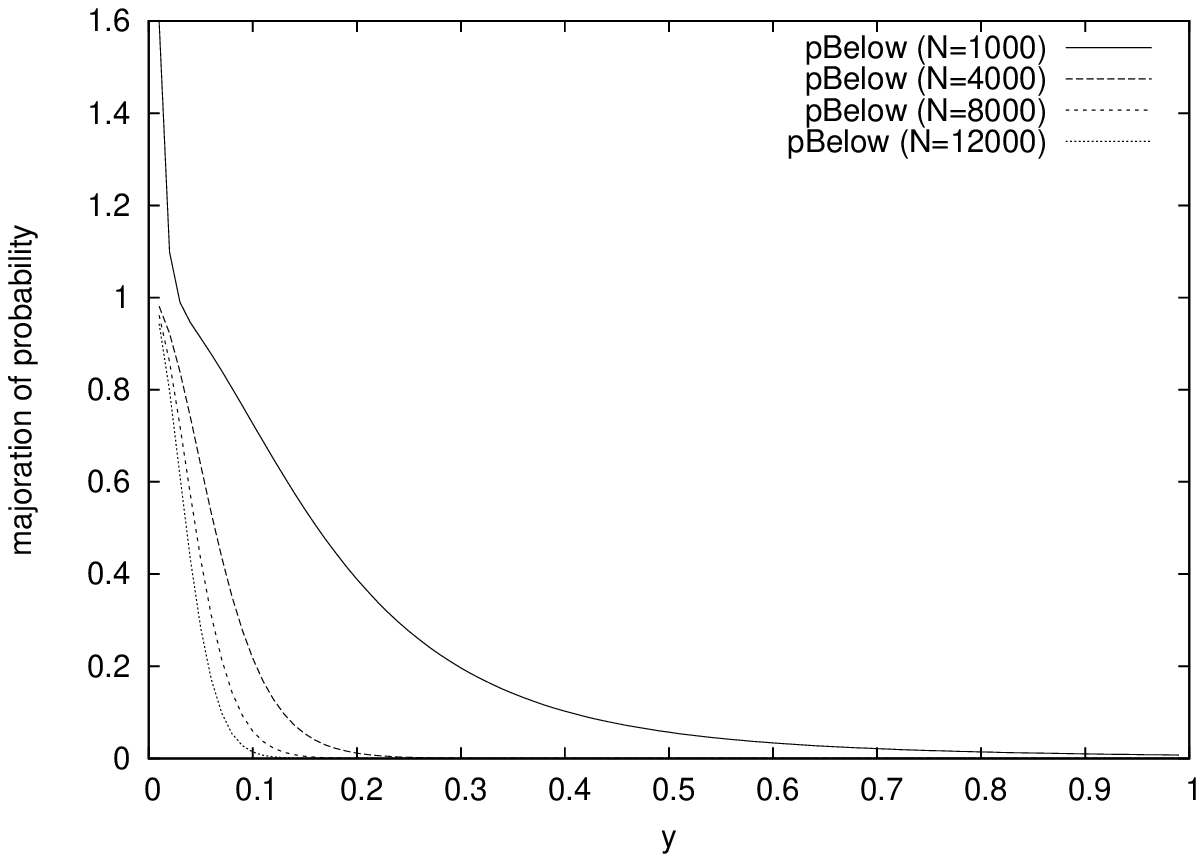} \\
		\includegraphics[scale=.6]{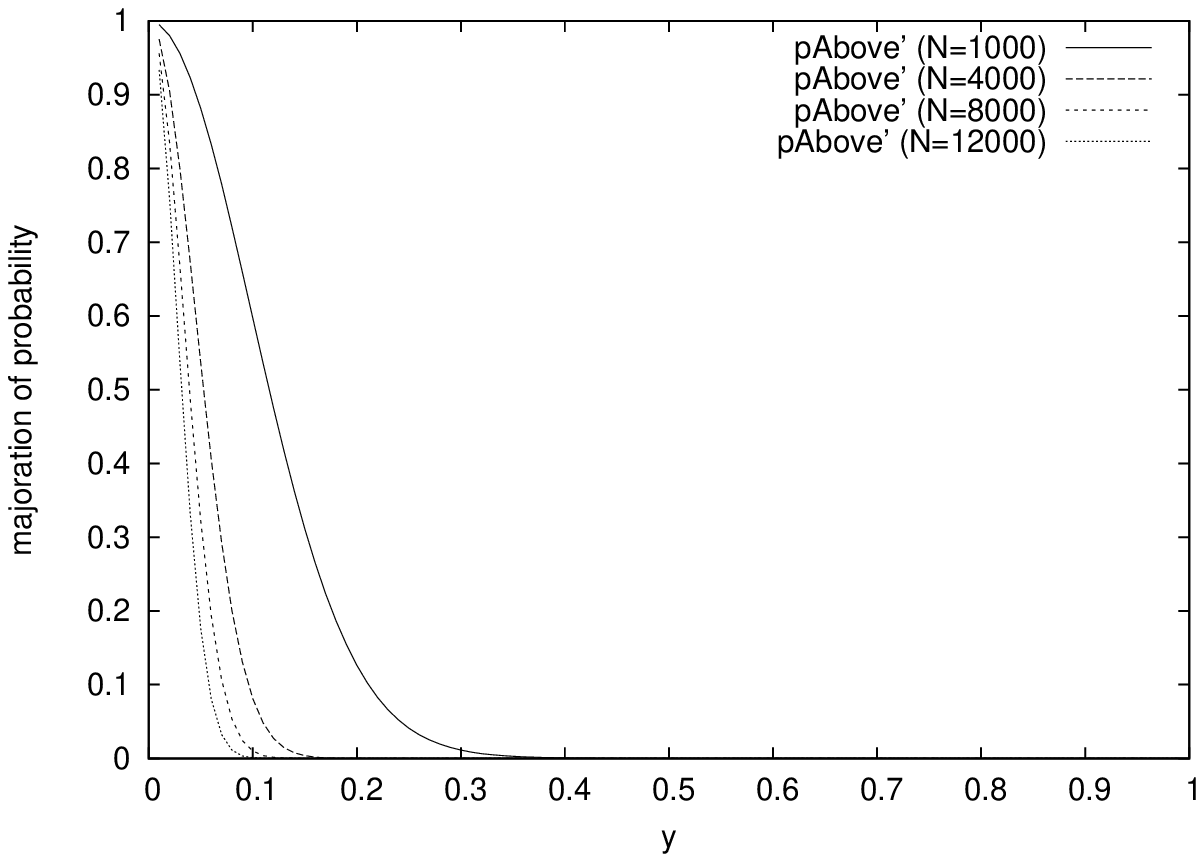}
		\includegraphics[scale=.6]{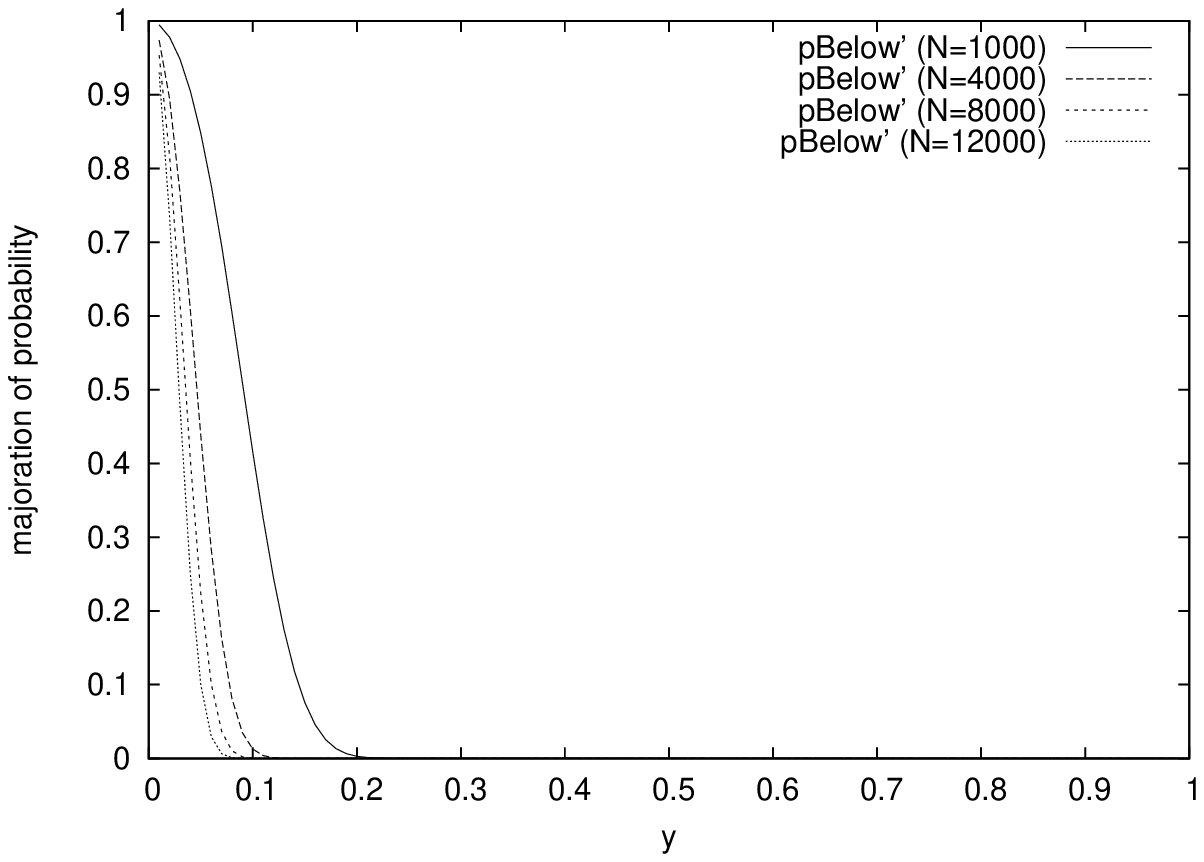}
		\caption{Plots, for $N\in\{1000, 4000, 8000, 12000\}$, of $pAbove(y,N)$ (left-top) and $pBelow(y,N)$ (right-top), $pAbove'(y,N)$ (left-bottom) and $pBelow'(y,N)$ (right-bottom) for the Ishigami model and for ${\bf u}=\{1\}$. }
		\label{fig:concentre}
	\end{center}
\end{figure}

As expected, the concentration inequalities are more conservative than the asymptotic confidence interval. These plots confirm that the $T^{\bf u}_{N,\mathrm{Cl}}$ concentrates faster than $ S^{\bf u}_{N,\mathrm{Cl}}$, and the inequality, while conservative, is sharp enough for this desirable property of $T^{\bf u}_{N,\mathrm{Cl}}$ to be reported. We also notice that there is a dissimetry in the bounds for above and below deviations, as this is often the case for concentration inequalities. Finally, the expected convergence for $N\rightarrow+\infty$ is observed. 

\section{Berry-Esseen Theorems}
\label{s:berry}

In this section we will give a general Berry-Esseen type Theorem for the estimator $S^{\bf u}_{N,\mathrm{Cl}}$ in one dimension (i.e. $k=1$). 
Let $\Phi$  be the cumulative distribution function of the standard Gaussian distribution.

\subsection{Pinelis' Theorem}\label{secpinelis}

We first recall a general Berry-Esseen type theorem proved in \cite{pinelis2009berry}.
Let $(V_i)_{i \geq 1}$ a sequence of i.i.d. centered random variables in $\mathbb{R}^d$, for some $d \in \mathbb{N}^*$. 
Let $f$ some measurable function: $\mathbb{R}^d \rightarrow \mathbb{R}$ with $f(0)=0$ and such that:
\begin{equation}\label{pinelis}
\exists \, \varepsilon >0, \, \exists \, M_{\varepsilon}>0 \text{ s. t. }
|f(x)-L(x)|\leq \frac{M_{\varepsilon}}{2}||x||^2
\end{equation}
 where $L:=Df(0)$ is the Fr\'echet derivative of $f$ at point $0$.
 
 \begin{rem}\label{rempinelis}
 Remark that condition (\ref{pinelis}) is satisfied as soon as $f$ is twice continuously differentiable in a neighborhood of $0$.
 \end{rem}
 
\begin{theorem}[Corollary 3.7 in \cite{pinelis2009berry}]\label{thpin}
Take any $p \in \, (2,3]$. Assume (\ref{pinelis}) holds, 
$${\sigma}:=\sqrt{\mathbb{E}\left(L(V)^2\right)} >0\, ,$$
and $\left(\mathbb{E}\left(||V||^p\right)\right)^{1/p} < \infty$ where $\| \cdot \|$ denotes the euclidean norm on $\mathbb{R}^d$. Then for all $z \in \mathbb{R}$
\begin{equation}\label{pinelisbis}
\left| \mathbb{P}\left(\frac{f(\overline{V}_n)}{{\sigma}/\sqrt{n}} \leq z \right)-\Phi(z)\right| \leq \frac{\kappa}{n^{p/2-1}} \, ,
\end{equation}
where $\kappa$ above is a generic constant that depends only upon p.
\end{theorem}

\subsection{Theoretical result for the general case}

For any random variable $Z$, denote by $Z^c$ its centered version $Z-\E(Z)$. 
\begin{theorem}
Assume that the random variable $Y$ has finite moments up to order $6$. Then,
for all $z \in\mathbb{R}$,
\begin{equation}
\left|\P\left(\frac{\sqrt{N}}{\sigma}\left[S^{\bf u}_{N,\mathrm{Cl}}-\myS^{{\bf u}}\right]\leq z\right)-\Phi(z) \right|\leq \frac{\kappa}{\sqrt{N}} \, .
\label{berichon1}
\end{equation}
Here
\begin{equation*}
\sigma^2:=\Var \left(\frac{1}{V} \left(Y^c(Y^{\bf u})^c-S^{\bf u}_{\mathrm{Cl}}(Y^c)^2\right)\right)
\label{berichon2}
\end{equation*}
is the asymptotic variance of $\sqrt{N} S^{\bf u}_{N,\mathrm{Cl}}$. 
\end{theorem}
\begin{proof}
We define 
$V_i=\left(Y_i^c(Y_i^{\bf u})^c-C_u,Y_i^c,(Y_i^{\bf u})^c,(Y_i^c)^2-V\right)^t$ and $f~: \mathbb{R}^4 \rightarrow \mathbb{R}$ as $f(x,y,z,t)=\frac{x-yz+C_u}{t-y^2+V}-\myS^{{\bf u}}$. Note that $f(0,0,0,0)=0$, $f(\overline{V}_N)=S^{\bf u}_{N,\mathrm{Cl}}-\myS^{{\bf u}}$ and by Remark \ref{rempinelis}, (\ref{pinelis}) holds. The result is then a direct application of Theorem \ref{thpin} once  
$\sigma^2=\mathbb{E}\left(L(V)^2\right)$ will be computed. We have 
$$\left(\frac{\partial f}{\partial x},\frac{\partial f}{\partial y},\frac{\partial f}{\partial z},\frac{\partial f}{\partial t}\right)(0,0,0,0)=\left(\frac1V,0,0,\frac{-C_u}{V^2}\right)\,.$$ 
Using notation in Section \ref{secpinelis} one gets
$$L(x,y,z,t)=\frac{1}{V}\left(x- \myS^{{\bf u}} t\right)\, \quad \trm{and} \quad L(V)=\frac{1}{V}\left(Y^c(Y^{\bf u})^c-\myS^{{\bf u}}(Y^c)^2\right)\, .$$
Straightforward computations lead to the required result.
\end{proof}

Then we have a Berry-Essen theorem for Sobol index estimator in a general case (whatever the first moment of $Y$). However, the constant of the bound is hard or even too complex to express explicitely. In the next section we present a Berry-Essen theorem with explicit bounds in the centered case but with an estimator of $S^{\bf u}_{\mathrm{Cl}}$ slightly different from $S^{\bf u}_{N,\mathrm{Cl}}$.

\subsection{Practical result in the centered case}
In this section we give a Berry-Esseen theorem for the estimator 
$$\widetilde{S}^{\bf u}_{N,\mathrm{Cl}}:=\frac{ \frac{1}{N} \sum  Y_i   Y_i^{\bf u}}{ \frac{1}{N}\sum  Y_i^2 }$$ 
in the centered case and $k=1$.  Further, let $\kappa\approx 0.42$ be the last best constant known in the classical Berry-Esseen theorem (\cite{popov}). We then have 

\begin{theorem}
Assume that the random variable $Y$ has finite moment up to order $6$. Then,
for all $t\in\mathbb{R}$,
\begin{equation}
\left|\P\left(\frac{\sqrt{N}}{\sigma}(\widetilde{S}^{\bf u}_{N,\mathrm{Cl}}-S^{\bf u}_{\mathrm{Cl}})\leq t\right)-\Phi(t) \right|\leq \frac{\kappa\mu_{3,N}}{\sqrt{N}}+\left|\Phi(t)-
\Phi\left(\frac{t}{\sqrt{1+\frac{t\nu_N}{\sigma\sqrt{N}V^2}}}\right)\right|.
\label{berichon1Centre}
\end{equation}
Here
\begin{equation}
\sigma^2:=\Var \left(\frac{1}{V} \left(YY^{\bf u}-S^{\bf u}_{\mathrm{Cl}}Y^2\right)\right)
\label{berichon2bis}
\end{equation}
is the asymptotic variance of $\sqrt{N} \widetilde{S}^{\bf u}_{N,\mathrm{Cl}}$ and 
\begin{eqnarray*}
\mu_{3,N}&:=&\E\left[\left|\frac{\Delta_n-\E(\Delta_n)}{\sqrt{\Var \Delta_n}}\right|^3\right],\\
\Delta_N& :=& \sigma^{-1}V\left[YY^{\bf u}-\left(S^{\bf u}_{\mathrm{Cl}} +\frac{t\sigma}{\sqrt{N}}\right)Y^2\right],\\
\nu_N&:=&\left(\frac{t\sigma}{\sqrt{N}}+2S^{\bf u}_{\mathrm{Cl}}\right)\Var (Y^2)-2\Cov(YY^{\bf u},Y^2).
\end{eqnarray*}
\end{theorem}
\begin{proof}

To begin with, we compute the asymptotic variance $\sigma^2$ of $\sqrt{N}\widetilde{S}^{\bf u}_{N,\mathrm{Cl}}$ :
we apply the so-called Delta method \cite{van2000asymptotic} to $W_i=(Y_iY_i^{\bf u},Y_i^2)$ and $\Psi(u,v):=uv^{-1},\;(u\in\mathbb{R},v>0)$. Then  $\sigma^2=J_{\Psi}(\E(W))\Sigma J_{\Psi}(\E(W))^t$ ($J_{\Psi}$ the Jacobian of $\Psi$) and the expression given in (\ref{berichon2bis}) follows obviously.

Now, for $t\in\mathbb{R}$, set 
$$A_t:=\left\{\frac{\sqrt{N}}{\sigma}\left(\widetilde{S}^{\bf u}_{N,\mathrm{Cl}}-S^{\bf u}_{\mathrm{Cl}}\right)\leq t \right\}.$$
Obvious algrebraic manipulations lead to
$$A_t=\left\{\sqrt{N^{-1}}\sum_{j=1}^{N}\Delta_{N,j}\leq 0\right\}$$
where, for $j=1,\ldots,N$,
$$\Delta_{N,j}:=\sigma^{-1}\left[Y_jY^{\bf u}_j V-\left(C_u +\frac{t\sigma}{\sqrt{N}}V\right)Y_j^2\right].$$

Now, we have
\begin{eqnarray*}
\E(\Delta_N)& =& \frac{-tV^2}{\sqrt{N}} \quad \trm{and} \quad \Var( \Delta_N) = V^4\left[1+ \frac{t\nu_N}{\sigma\sqrt{N}V^2}\right].
\end{eqnarray*}
%
%
%

So that denoting by $\Delta_{N,\cdot}$ the empirical mean of $(\Delta_{N,j})_{j=1,\ldots,N}$, we obtain
$$A_t=\left\{\sqrt{N}\Delta_{N,\cdot}\leq 0\right\}=\left\{\sqrt{N}\left(\frac{\Delta_{N,\cdot}-\E(\Delta_N)}{\sqrt{\Var\Delta_N}}\right)\leq 
\frac{t}{\sqrt{1+\frac{t\nu_N}{\sigma\sqrt{N}V^2}}}\right\}.$$

%
%

Now, to conclude we apply Berry-Esseen theorem (see \cite{popov}) and the triangular inequality to obtain (\ref{berichon1}).\qedhere
\end{proof}

\subsection{Numerical applications for the centered case}

We denote by $B(t)$ the right hand side of the Berry-Esseen inequality \eqref{berichon1Centre}. It is clear that, for any $y>0$, we have:
\begin{equation}
\label{e:BECI}
\P( -y \leq \widetilde{S}^{\bf u}_{N,\mathrm{Cl}} - S^{\bf u}_{\mathrm{Cl}} \leq y) \geq
\left[ \Phi\left(\frac{\sqrt N}{\sigma}y\right)-\Phi\left(-\frac{\sqrt N}{\sigma} y\right) \right] - 
\left[ B\left(\frac{\sqrt N}{\sigma}y\right)+B\left(-\frac{\sqrt N}{\sigma}y\right) \right] 
\end{equation}
and:
\begin{equation}
\label{e:BECIUpper}
\P( -y \leq \widetilde{S}^{\bf u}_{N,\mathrm{Cl}} - S^{\bf u}_{\mathrm{Cl}} \leq y) \leq
\left[ \Phi\left(\frac{\sqrt N}{\sigma}y\right)-\Phi\left(-\frac{\sqrt N}{\sigma} y\right) \right] + 
\left[ B\left(\frac{\sqrt N}{\sigma}y\right)+B\left(-\frac{\sqrt N}{\sigma}y\right) \right] 
\end{equation}

Hence, the actual confidence level of the asymptotic confidence interval for $S^{\bf u}$ using $\widetilde S^{\bf u}_{N,\mathrm{Cl}}$ is greater than the theoretical level (first term of the sum above), minus a correction term given by the Berry-Esseen theorem (second term). The upper bound given by \eqref{e:BECIUpper} may also be of practical interest: an overly conservative (overconfident) interval is not always desirable, as a more precise interval with accurate level may exist. 

As in the previous applicational section \ref{ss:numconcentre}, the lower bound of the asymptotic confidence interval level involve unkown quantities (moments of $\Delta_N$, $Y$, $Y Y_u$) that have to be estimated. We designate by $L(y,N)$ (resp. $U(y,N)$) the estimator of the right hand side of \eqref{e:BECI} (resp. \eqref{e:BECIUpper}) when all unkown quantities are empirically estimated.

We take as output model the Ishigami function defined at Section \ref{ss:ishigamifun}, recentered by its true mean $7/2$ :
\[ Y=f(X_1,X_2,X_3)=\sin X_1+7 \sin^2 X_2+0.1 X_3^4 \sin X_1 - \frac{7}{2}. \]
Note that the true mean could also be replaced by an estimate of the mean. For $y$, we choose $y=1.96 \frac{\widehat{\sigma^2}}{\sqrt N}$, where $\widehat{\sigma^2}$ is an empirical estimate of $\sigma^2$, so as to compute  (estimators of ) upper and lower bounds of the actual level of the $95\%$-level confidence interval.

We present the numerical results, as functions of $N$, and for ${\bf u}=\{1\}$ in Figure \ref{f:berryEss}; for ${\bf u}=\{2\}$ or ${\bf u}=\{3\}$, the results were very similar.

\begin{figure}[h]
	\begin{center}
		\includegraphics[scale=.8]{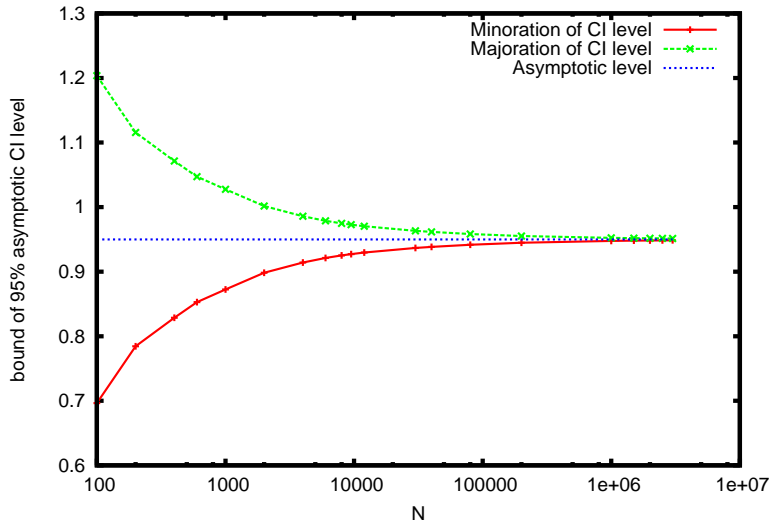}
		\caption{Plots of $L(N)$ (minoration of CI level)) and $U(N)$ (majoration of CI level) for ${\bf u}=\{1\}$ and different values of $N$. }
		\label{f:berryEss}
	\end{center}
\end{figure}

As expected, the actual confidence level is estimated under the ``target'' level of the confidence interval (0.95). As $N \rightarrow +\infty$, our bound converges (quite slowly) to 0.95. Nevertheless, the Berry-Esseen bound we have presented quickly attains confidence levels which are very close to the asymptotic level, and it can be used so as to provide a certification, at finite sample size, of the level of the asymptotic confidence interval.

\bigskip

\textbf{Acknowledgements. } This work has been partially supported by the French National
Research Agency (ANR) through COSINUS program (project COSTA-BRAVA
nr. ANR-09-COSI-015).

\bibliographystyle{plain}
\bibliography{biblio,bibi}
\end{document}